\RequirePackage{fix-cm}
\documentclass[smallcondensed]{svjour3}     
\smartqed  
\usepackage{graphicx}
\usepackage{amsthm,amsfonts}
\usepackage{amsmath}
\usepackage{amssymb}
\usepackage{graphicx}
\usepackage{makeidx}
\usepackage{enumitem}
\usepackage{color}

\usepackage{bm}

\newtheorem{mydef}{Definition}
\newtheorem{myteo}{Theorem}
\newtheorem{mypro}{Proposition}
\newtheorem{mycor}{Corollary}

\newtheorem{myexe}{Example}

\begin{document}

\title{A Distance Between Channels: the average error of mismatched channels
}


\author{Rafael G. L. D'Oliveira         \and
        Marcelo Firer 
}


\institute{Rafael G. L. D'Oliveira \at
              Illinois Institute of Technology \\
              \email{rgldoliveira@gmail.com}           
           \and
           Marcelo Firer \at
              University of Campinas\\
              \email{mfirer@ime.unicamp.br}
}

\date{Received: date / Accepted: date}

\maketitle

\begin{abstract}
Two channels are equivalent if their maximum likelihood (ML) decoders coincide for every code. We show that this equivalence relation partitions the space of channels into a generalized hyperplane arrangement. With this, we define a coding distance between channels in terms of their ML-decoders which is meaningful from the decoding point of view, in the sense that the closer two channels are, the larger is the probability of them sharing the same ML-decoder. We give explicit formulas for these probabilities.

\end{abstract}

\section{Introduction}

A communication channel cannot generally be chosen in application, but is rather considered to be a ``fact of life''. The most that is possible is to make measurements to characterize the type of noise and to have a model for errors. 

Some channels are simpler to be handled than others. If a channel is metrizable, for example, one can use methods from classical coding theory which make use of the metric\footnote{For the use of general distances in coding theory see \cite{deza14,gabi12}.}. If, furthermore, the metric is translation invariant one can use syndrome decoding which greatly reduces decoding complexity.

The question of being metrizable is one that underlines many aspects in coding theory, but it is seldom stated in an explicit way, so we do it here: a channel with equal input and output set of messages is said to be \textit{metrizable} if there is a metric such that, for any code, using maximum likelihood or minimum distance leads to the same decoding decisions\footnote{For a deeper look into conditions for metrization see references  \cite{doli16a,doli16d,doli16b,fire14} and \cite{popl16,qure16,segu80}.}.

A metric structure is just one kind of structure that makes a channel more manageable, and sometimes it may be worth to consider an alternative channel model which is less accurate as a model for noise and errors  but is simpler to manage in some sense, like for example with the existence of efficient decoding algorithms. In this sense, the long term goal is to develop an approximation-theory-like approach to coding theory. To do so, the first step would be to determine a distance in the space of channels which translates the probabilistic structure of channels. This is the main goal of this work.

\bigskip

This paper is organized in the following way:

In Section \ref{sec:related work} we discuss the connections between the work presented here with previous work on mismatched decoding and the partial ordering of channels.

In Section \ref{sec:pre} we establishing the notation and basic definitions used in this work. Among these is the notion of \emph{decoding equivalence} (Definition \ref{def:decequ}), a natural equivalence relation between channels. When speaking of the ``space of channels'' we  mean the set of channels under decoding equivalence.

In Section \ref{sec:decequirn} we show that the space of channels has a structure of a special kind of hyperplane arrangement known as the braid arrangement (Theorems \ref{teo:decequibra} and \ref{teo:braidcha}). 

In the literature, a braid arrangement has a natural distance (the Kendall tau distance), but this does not attend our requirements. 

In Sections \ref{sec:disper} and \ref{sec:discha} we present a modified version of the Kendall tau distance which is an appropriate measure from the decoding point of view: channels which are closer with respect to this distance are more probable to perform the same maximum likelihood decoding when considering arbitrary random codes (our main result, Corrolary \ref{cor:disrad}).

\section{Related Work} \label{sec:related work}

The study of the space of channels at its own sake is related (although not equivalent) to other subjects that has been studied, namely mismatched decoding and partial ordering of channels. We give a brief view of these topics, pointing the similarities and differences with our approach.

\subsection{Mismatched Decoding}

Our approximation-theory-like approach is similar to the setting of \emph{mismatched decoding}. In this setting instead of using the Maximum Likelihood decoder determined by the channel $P$ ($ML_{P}%
$-decoding), a different decoding criterion is used. In practice, this might occur due to inaccuracies in the measurement of a channel.  In this case we are using an
$ML_{Q}$-decoding, where $Q$ is the non-accurate measured channel. Another
reason for mismatched decoding arises when there are no reasonable algorithms
for implementing $ML_{P}$-decoding. 

Mismatched decoding has an extensive literature (\cite{gant00} has many relevant references on the subject). The approach, however, is essentially information theoretical,
guided by the fundamental question of determining \emph{what can, in principle, be done}. This means that most of the work in the area aims to understand what is achievable asymptotically, for example, what are the achievable rates for families of channels
with the input-output sets' size going to infinity. Those are very difficult
questions and hence a significant part of the effort is directed to find bounds
for those rates (and other significant invariants). 

Our approach is less concerned with the asymptotic aspects of achievability. Once the input and output sets $X$ and $Y$ are given (and fixed) and supposing that the actual channel is
$P$, how much are we expected to loose once we decode a randomly chosen code
using the ML-decoding criterion determined by a different channel $Q$. Our measure of expected
loss is the overall probability of error in the whole process of encoding,
transmitting, and decoding. In this sense, we may say that we are considering
the mismatched decoding problem in the finite block length regime.

\subsection{Partial Ordering of Channels}

Our approach to study the geometry of the space of channels has an
intersection with the concept of \emph{channel inclusion}, as introduced by
Shannon \cite{shan58} and as presented, for example, by Makur and Polyanskiy \cite{maku16}.
Using the notation of Makur and Polyanskiy, given two channels with transition
matrices $P$ and $Q$ of size $N_{P}\times M_{P}$ and $N_{Q}\times M_{Q}$
respectively, with $N_{Q}\leq N_{P}$ and $M_{Q}\leq M_{P}$, one says that $P$
\emph{includes }$Q$ if there are two families $\left(  A_{k}\right)
_{k=1}^{m}$ and $\left(  B_{k}\right)  _{k=1}^{m}$ of channels (with $A_{k}$
being an $M_{P}\times M_{Q}$ transition matrix and $A_{K}$ an $N_{Q}\times
N_{P}$ transition matrix) and probability mass function $g$ over the set
$\left\{  1,2,...,m\right\}  $ such that
\[
Q=\sum_{k=1}^{m}g\left(  k\right)  B_{k}PA_{k}\text{.}%
\]

This concepts embraces many different situations, some of which can be
understood with our definition of the space of channels with the decoding
equivalence. For example, the first example introduced in Figure 1 of Shannon's
work corresponds to the situation where $m=1$ and both $A_{1}$ and $B_{1}$ is
determined by a projection matrix. If we allow $A_{1}$ and $B_{1}$ to correspond
to a projection or a permutation matrix (or a combination of both), we
actually have a hyperplane, $\mathcal{P}_{\pi}\subset\mathbb{R}^{N\times M}$,
with a braid arrangemente structure (see Section \ref{sec:decequirn}) induced from the braid
arrangement structure of $\mathbb{R}^{N\times M}$, by considering the
intersection of a deconding cone $cone\left(  x\right)  \cap\mathcal{P}$.

\section{Preliminaries} \label{sec:pre}
In this section we start with a list of definitions and notations used throughout this work. Since these concepts are well known we present them  very succinctly, citing references for details. After that, in \ref{subsec:space_of_chan}, we present the decoding equivalence between channels introduced in \cite{doli16a} and define the space of channels.

\subsection{Notation}\label{cones}

We consider the basic setting of information theory \cite{shan48} where a transmitter sends a message to a receiver passing through a channel. Let $X = \{x_1, x_2, \ldots, x_n\}$ be the set of input messages which the transmitter can send and let $Y = \{y_1, y_2, \ldots, y_m\}$ be the set of output messages which the receiver can receive. It is common for the messages to come from some alphabet in which case the sets $X$ and $Y$ are exponential on the block length with respect to the size of the alphabet.

 A \emph{channel} is a $n \times m$ probabilistic matrix $P$ such that $P_{ij} = Pr (y_j \hspace{2pt} \text{received} \hspace{3pt} \vert x_i \hspace{2pt} \text{sent})$, the probability of receiving $y_j$ given that $x_i$ was sent (the rows sum to $1$).

Given a code $\emptyset \neq C \subseteq X$, a \emph{maximum likelihood decoder} is such that $y \in Y$ is decoded as some $c \in C$ which maximizes $ Pr (y \hspace{2pt} \text{received} \hspace{3pt} \vert c \hspace{2pt} \text{sent})$. The set of maximum likelihood (ML) decoders of the  channel $P$ for a code $C$ is denoted by ${D\widehat{e}c}_C (P)$.

A \emph{weak order} over a set $X$ is a triple $(X,\prec,\simeq)$, where $\prec$ and $\simeq$ are binary relations on $X$ satisfying, for all $x,y,z \in X$:
\begin{enumerate}
\item $x \prec y$ and $y \prec z$ implies that $x \prec z$,
\item $\simeq$ is an equivalence relation,
\item exactly one of $x \prec y$, $y \prec x$  or $x \simeq y$ holds.
\end{enumerate}
We denote the set of all weak orders over $n$ objects by $W_n$.

We denote the symmetric group over $n$ objects by $S_n$. As usually done, we use lowercase Greek letters for elements of this set ($\sigma, \tau, \phi \in S_n$). 

A set $ A \subseteq \mathbb{R}^n$ is \textit{convex} if it contains the segment joining any two of its points, i.e. $\alpha x + (1-\alpha)y \in A$ for every $x,y \in A$ and $0 \leq \alpha \leq 1$.

A \textit{hyperplane }is a set $H \subseteq \mathbb{R}^n$ of the form $H = \{ x \in \mathbb{R}^n : \alpha \cdot x  = a \}$ where $0 \neq \alpha \in \mathbb{R}^n$, $a \in \mathbb{R}$ and $\alpha \cdot x := \sum_{i=1}^n \alpha_i x_i $ is the usual dot product.

A \textit{hyperplane arrangement}  $\mathcal{A}$ (see \cite{stan04} for details) is a set of hyperplanes. A region of an arrangement is a connected component of the complement of the hyperplanes, $X = \mathbb{R}^n - \bigcup\limits_{H \in \mathcal{A}} H$. The set of regions is denoted by $\mathcal{R}(\mathcal{A})$ and $r(\mathcal{A}) := \# \mathcal{R}(\mathcal{A})$.

Each hyperplane divides $\mathbb{R}^n$ into two subsets known as \textit{half-spaces}. The two half spaces corresponding to $H = \{ x \in \mathbb{R}^n : \alpha \cdot x  = a \}$ are $\{ x \in \mathbb{R}^n : \alpha \cdot x  \leq a \}$ and $\{ x \in \mathbb{R}^n : \alpha \cdot x  \geq a \}$.

A \textit{convex polytope} is the intersection of a finite  set of half-spaces which is bounded.

A set $C \subseteq \mathbb{R}^n$ is a  \textit{convex cone }if $\alpha x + \beta y \in C$ for every $x,y \in C$ and $\alpha,\beta \geq 0$.

We are particularly interested in the\textit{ braid arrangement,} $\mathcal{B}_n$, which consists of the $\binom{n}{2}$ hyperplanes: $x_i - x_j = 0$ for $1 \leq i < j \leq n$. Specifying to which side of the hyperplane a point $a \in \mathbb{R}^n$ belongs to is equivalent to determining whether $a_i < a_j$ or $a_j < a_i$. Doing so for every hyperplane is equivalent to imposing a linear order on the $a_i$. So to each permutation $\sigma \in S_n$ there corresponds a region $R_{\sigma} \in \mathcal{R}(\mathcal{B}_n)$ given by $R_{\sigma} = \{ x \in \mathbb{R}^n : a_{\sigma (1)} < a_{\sigma (2)} <\ldots < a_{\sigma(n)} \}$. Thus, $r(\mathcal{B}_n) = n!$.

The Iverson bracket will be used in our definitions and proofs: for a statement $S$, the bracket $[S]$ equals $1$ if statement $S$ is true and equals $0$ otherwise.

\subsection{The Space of Channels}\label{subsec:space_of_chan}

The results in this section appear in more detail in \cite{doli16a},\cite{doli16d} and \cite{doli16b}.

Consider the space $\mathbb{R}_{\geq 0}^{n\times m}$ of matrices with non-negative entries. The space of all $n\times m$ channels, ${Ch}_{n\times m}$, is a subset of this space.

\begin{mydef} \label{def:decequ}
	Two channels $P,Q \in {Ch}_{n\times m}$ are decoding equivalent, $P \sim Q$, if, for any code $C \subseteq X$, they have the same maximum likelihood decoders, i.e. for every $C\subseteq X $, ${D\widehat{e}c}_C (P) = {D\widehat{e}c}_C (Q)$.
\end{mydef}

Our next definition will help characterize decoding equivalence.

\begin{mydef}
	Given a matrix $M \in \mathbb{R}_{\geq 0}^{n\times m}$, its \emph{weak order matrix} is the matrix $O^-M$ such that ${(O^-M)}_{ij} = k$ if $M_{ij}$ is the $k$-th largest element (allowing ties) in the $j$-th column of $M$.
\end{mydef}

\begin{myexe}
	If
	$ M = 
	\begin{pmatrix}
	9 & 2 & 1 \\ 
	9 & 7 & 0 \\ 
	8 & 6 & 8
	\end{pmatrix}
	$,
	then
	$O^- M = 
	\begin{pmatrix}
	1 & 3 & 2 \\ 
	1 & 1 & 3 \\ 
	2 & 2 & 1
	\end{pmatrix}
	$.
\end{myexe}

\begin{mypro}
	Two channels $P,Q \in {Ch}_{n\times m}$ are decoding equivalent if and only if $O^-P = O^-Q$.
\end{mypro}

\begin{proof}
Corollary $3$ in \cite{doli16a}.
\end{proof}

With this the decoding equivalence can be extended to the whole of $\mathbb{R}_{\geq 0}^{n\times m}$ by defining $M \sim N$ if $O^-M = O^-N$.

The decoding equivalence partitions $\mathbb{R}_{\geq 0}^{n\times m}$ into $\vert W_n \vert^m$ cones, $(n!)^m$ of which are full dimensional. We denote the decoding cone containing a matrix $M$ by $Cone(M)$ and note that they are the fibers of $O^-$, i.e.  $Cone = (O^-)^{-1}\circ O^-$. For details see \cite[Section 3]{doli16b}.

As we shall see, the space of channels has a structure of hyperplane arrangements and the simplicial structure of hyperplane arrangements reflects the structure of maximum likelihood decoding.

\section{Decoding Equivalence and the Braid Arrangement} \label{sec:decequirn}

Maximum likelihood decoding is done comparing  entries of a column of a matrix, the column corresponding to the received message.  Considering a column as a vector $x\in\mathbb{R}^n_{\geq 0}$, we show that the decoding equivalence partitions $\mathbb{R}^n_{\geq 0}$ into generalized regions of the braid arrangement. We then extend this result to $\mathbb{R}^{n \times  m}_{\geq 0}$.

We first define the $Order$ function.

\begin{mydef} \label{def:orderfunction}
The $Order$ function, $Order:\mathbb{R}^n_{\geq 0} \rightarrow W_n$, takes a vector $x \in \mathbb{R}^n_{\geq 0}$ to the weak ordering of its coordinates.
\end{mydef}

So, for example, $Order (\sqrt{2},\frac{-1}{2},\sqrt{2}) = Order (2,1,2) = (2 \prec 1 \simeq 3)$.

\begin{mypro} \label{pro:decequiord}
Two vectors $x,y \in \mathbb{R}_{> 0}^{n}$ are decoding equivalent if and only if $Order (x) = Order (y)$.
\end{mypro}

\begin{proof}
This follows because $Order (x) = Order (y)$ if and only if $O^- x = O^- y$.
\end{proof}

The \textit{fibers} of the $Order$ function, i.e., the inverse images $Order^{-1}(y)$ , partition $\mathbb{R}^n$ into the decoding equivalence class.

\begin{mydef}
The cone function is given by $Cone : \mathbb{R}^n_{\geq 0} \rightarrow 2^{\mathbb{R}^n_{\geq 0}}$ such that $Cone (x) = ({Order})^{-1} \circ Order$. We call $Cone(x)$ the \emph{decoding cone of} $x$.
\end{mydef}
 It is clear, from the definition, that two channels in the same decoding cone determine the same maximum likelihood criteria, for every code.
 
We generalize the definition of the region of a hyperplane arrangement.

\begin{mydef} \label{def:genreg}
A generalized region of a hyperplane arrangement $\mathcal{A}$ is a connected component of $ 
\bigcap\limits_{H \in \mathcal{A}_1} H - \bigcup\limits_{H \in \mathcal{A}_2} H$, where $\mathcal{A}_1,\mathcal{A}_2$ is a disjoint partition of $\mathcal{A}$. We denote the sets of generalized regions by $\mathcal{GR}(\mathcal{A})$ and $gr(\mathcal{A}) = \# \mathcal{GR}(\mathcal{A})$.
\end{mydef}

 As stated in Section \ref{cones}, the braid arrangement consists of the $\binom{n}{2}$ hyperplanes: $H_{ij} = \{ x \in \mathbb{R}^n : x_i=x_j \}$ for $1 \leq i <j \leq n$. The next theorem shows that the decoding equivalence partitions $\mathbb{R}^n$ into generalized regions of the braid arrangement.

\begin{myteo} \label{teo:decequibra}
Let $x,y \in \mathbb{R}^n_{\geq 0}$. Then, $x$ is decoding equivalent to $y$ if and only if $x,y \in R$ for some $R \in \mathcal{GR}(\mathcal{B}_n)$, where $\mathcal{B}_n$ is the braid arrangement.
\end{myteo}

\begin{proof}
Specifying to which generalized region $R_x \in \mathcal{GR}(\mathcal{B}_n)$ a point $x \in \mathbb{R}^n$ belongs to is equivalent to determining whether $x_i < x_j$, $x_i = x_j$ or $x_i > x_j$ for every $1 \leq i <j \leq n$. This is equivalent to imposing a weak order on the coordinates of $x$. But this implies that $y \in R_x$ if and only if $Order(y) = Order(x)$. The result then follows from Proposition \ref{pro:decequiord}.
\end{proof}

In other words, if $R \in \mathcal{GR}(\mathcal{B}_n)$ then $x \in R$ if and only if $R = Cone(x)$, i.e. the decoding cones are the generalized regions of the braid arrangement.

\begin{figure}[h] \label{fig:cone}
\centering
\includegraphics[scale=0.2]{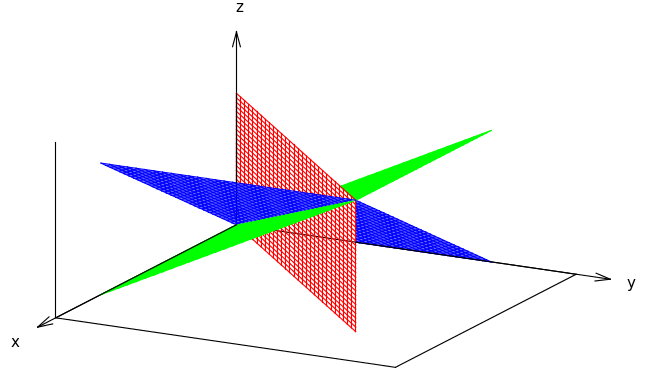}
\caption{The partition of $\mathbb{R}_{> 0}^3$ by decoding equivalence into $13$ cones: six $3$-dimensional, six $2$-dimensional, and one $1$-dimensional (the ray  $(\lambda,\lambda,\lambda)$ with $\lambda > 0)$.}
\end{figure}

We now extend the results from $\mathbb{R}^{n}$ to $\mathbb{R}^{n \times m}$.

\begin{mydef} \label{def:ordmat}
The $Order$ function, $Order:\mathbb{R}^{n \times m}_{\geq 0} \rightarrow W_n^m$, is defined as 
\[ Order (M) = Order (M[\cdot][1]) \times  Order (M[\cdot][2]) \times  ... \times  Order (M[\cdot][m]) ,\]
where $Order (M[\cdot][j])$ is the the order function in Definition \ref{def:orderfunction} applied to the $j$-th column of $M$. The decoding cone of $M$ is $Cone(M) = {Order}^{-1} \circ {Order} (M)$.
\end{mydef}

The following result is an analog of Theorem \ref{teo:decequibra}.

\begin{myteo} \label{teo:braidcha}
Let $M,M' \in \mathbb{R}^{n \times m}$. Then, $M$ is decoding equivalent to $M'$ if and only if $M[\cdot][j],M'[\cdot][j] \in R_j$ for some $R_j \in \mathcal{GR}(\mathcal{B}_n)$, where $\mathcal{B}_n$ is the braid arrangement.
\end{myteo}

\begin{proof}
The proof is equivalent to that of Theorem \ref{teo:decequibra} by using Definition \ref{def:ordmat}.
\end{proof}

\section{A Decoding Distance Between Permutations} \label{sec:disper}

Having an appropriate model of the transmission channel is not always good enough to establish all the necessities in the communication process. Many other questions, such as the complexity of the decoding algorithms, need to be taken into consideration. For this reason, for example, the Hamming metric is many times used, even when the channel is not the binary symmetric channel.

In this sense, it may be interesting to develop an ``approximation theory'' for channels. The idea is that we can use an approximate simpler channel (or distance matched to it) in place of the original one.

The most basic and mandatory tool for the development of an approximation theory is a distance in  the space $Cha_{n\times m}$ which is adequate in some sense. If $P$ is a channel, $P_y$ denotes the column corresponding to $y$ being a received message. We will propose a relevant distance on $Cha_{n\times m}$  which relates to the following:

\begin{center}
	
\emph{Let $P , Q \in {Cha}_{n\times  m}$ be two different channels and suppose we know what output $y \in Y$ is received. Choosing a code $C \subseteq X$ from the set of all codes with uniform distribution, what is the probability that ${D\widehat{e}c}_C (P_y) \cap {D\widehat{e}c}_C (Q_y) \neq \emptyset$?}

\end{center}

When we say that the distance is related to that question it means that the probability that ${D\widehat{e}c}_C (P_y) \cap {D\widehat{e}c}_C (Q_y) \neq \emptyset$ decreases with the purposed distance: the closer the channels are, the more probable they are to determine the same decoders.

Since we know what output $y$ is received, only its corresponding column matters for decoding. Thus we are dealing with the decoding equivalence in $\mathbb{R}^n$.

We will only consider the cases for which $Cone(P_y)$ and $Cone(Q_y)$ are $n$-dimensional (and leave the general case for future work). We say that a channel $P$ such that $Cone(P)$ is full dimensional is a \emph{stable channel}, since small perturbations of the channel probabilities do not affect the decoding decisions. In this case ${D\widehat{e}c}_C (P_y) \cap {D\widehat{e}c}_C (Q_y) \neq \emptyset$ is equivalent to ${D\widehat{e}c}_C (P_y) = {D\widehat{e}c}_C (Q_y)$.

By Theorem \ref{teo:decequibra}, each $n$-dimensional decoding cone corresponds to a region of the braid arrangement $\mathcal{B}_n$. As noted in Section \ref{cones} to each $\sigma \in S_n$ there corresponds a region $R_\sigma \in \mathcal{R}(\mathcal{B}_n)$. We can therefore identify every $n$-dimensional decoding cone with a permutation in $S_n$.

\begin{myexe}
Consider $\mathbb{R}_{\geq 0}^3$. The identity element $1 \in S_3$ corresponds to the cone with ordering $(1 \prec 2 \prec 3)$. The transposition $(13)\in S_3$ corresponds to the cone with ordering $(3 \prec 2 \prec 1)$.
\end{myexe}

Since decoding depends exclusively on the decoding cone, we can extend the definition of ${D\widehat{e}c}_C$ to permutations in the following way.

\begin{mydef}
Let $\sigma \in S_n$, $R_{\sigma} \in \mathcal{R}(\mathcal{B}_n)$ its corresponding decoding cone and $P \in {Cha}_{n\times m}$ such that $P \in R_{\sigma}$. We define ${D\widehat{e}c}_C (\sigma) = {D\widehat{e}c}_C (P)$ for every $C \subseteq X$.
\end{mydef}

The leading question we posed in the beginning of this section can now be restated in terms of permutation groups as follows:

\begin{center}
Given two permutations $\sigma, \phi \in S_n$, what is the probability that ${D\widehat{e}c}_C (\sigma) = {D\widehat{e}c}_C (\phi)$ if $C \subseteq X$ is chosen with uniform distribution?
\end{center}

More precisely, we are interested in computing the following distance:

\begin{mydef} \label{def:decdist}
The decoding distance between two permutations $\sigma, \phi \in S_n$ is 
\[ d_{dec} (\sigma, \phi) = 1 - Pr( \hspace{3pt}{D\widehat{e}c}_C (\sigma) = {D\widehat{e}c}_C (\phi)\hspace{3pt} ),\]  where a code $C\subseteq X$ is chosen randomly, with uniform distribution.
\end{mydef}

We will solve this by elementary counting.

\begin{mydef}
Let $\sigma, \phi \in S_n$. We denote by $S(\sigma, \phi)$ the number of codes $C$ for which ${D\widehat{e}c}_C (\sigma) = {D\widehat{e}c}_C (\phi)$.
\end{mydef}

 We aim to relate $S(\sigma, \phi)$ to $ d_{dec} (\sigma, \phi)$.   We first remark that  $\sigma, \phi \in S_n$ is invariant by permutations.

\begin{mypro}
Let $\sigma, \phi, \tau \in S_n$. Then, $S(\tau \circ \sigma, \tau \circ \phi) = S(\sigma, \phi)$.
\end{mypro}

\begin{proof}
This follows from the fact that if you permute the rows of a channel, the same permutation on a maximum likelihood decoder of it will yield a maximum likelihood decoder of the permuted channel.
\end{proof}

Thus, we can define $S(\sigma) = S(1,\sigma)$ and then $S(\sigma, \phi) = S(\phi^{-1} \circ \sigma)$.

We now show how to compute this function.

\begin{myteo}\label{f_i}
Let $\sigma \in S_n$ and let us define $f_i(\sigma) = \sum_{j=i+1}^n [ \sigma^{-1} (i) \leq \sigma^{-1} (j) ]$. Then, $$S(\sigma) = \sum_{i=1}^n 2^{f_i(\sigma)}.$$

\end{myteo}

\begin{proof}
We want to count how many codes such that ${D\widehat{e}c}_C (1) = {D\widehat{e}c}_C (\sigma)$. The identity element represents $(1\prec 2\prec \ldots \prec n)$ and $\sigma$ represents $(\sigma^{-1}(1)\prec \sigma^{-1}(2)\prec \ldots \prec \sigma^{-1}(n))$.

Recall that $\sigma$ corresponds to a channel $P$ with input messages $X= \{ x_1, \ldots, x_n \}$. Consider codes $C$ such that $x_1 \in C$. The identity element $1$ will decode any one of these as $x_1$. Thus ${D\widehat{e}c}_C (1) = {D\widehat{e}c}_C (\sigma)$ if and only if $\sigma$ also decodes as $x_1$. For this to happen, $C$ can only contain elements $x_i$ such that $\sigma^{-1} (i) \leq \sigma^{-1} (j)$. But $f_1(\sigma)$ counts precisely how many of these exist. So the total number of codes satisfying $x_1 \in C$ and ${D\widehat{e}c}_C (1) = {D\widehat{e}c}_C (\sigma)$ is $2^{f_1(\sigma)}$.

Now consider codes $C$ such that $x_1 \notin C$ and $x_2 \in C$. The same reasoning yields the total number of codes satisfying $x_1 \notin C$, $x_2 \in C$ and ${D\widehat{e}c}_C (1) = {D\widehat{e}c}_C (\sigma)$ as $2^{f_2(\sigma)}$.

Continuing  with the same argument yields our result.
\end{proof}

The next theorem answers the question posed in the beginning of this section.

\begin{myteo}

Let $\sigma, \phi \in S_n$. If a code $C \subseteq X$ is picked uniformly distributed from the space of all codes, then $Pr({D\widehat{e}c}_C (\sigma) = {D\widehat{e}c}_C (\phi)) = \frac{S(\phi^{-1} \circ \sigma)}{2^n - 1}$.

\end{myteo}

\begin{proof}
By definition, $S(\phi^{-1} \circ \sigma)$ counts the number of codes such that ${D\widehat{e}c}_C (\sigma) = {D\widehat{e}c}_C (\phi)$. Elementary probability says we must divide this by the total number of codes.
\end{proof}

As a direct corollary we compute the decoding distance.

\begin{mycor} \label{cor:decodingdist}
The decoding distance between two permutations $\sigma, \phi \in S_n$ is 
\[ d_{dec} (\sigma, \phi) = 1 - \frac{S(\phi^{-1} \circ \sigma)}{2^n - 1}\]
\end{mycor}

In the context of the braid arrangement there exists already a natural distance between permutations. It is known as the Kendall tau distance \cite{kend38}, which we denote by $d_{\tau} (\sigma, \phi)$, and is defined as the minimum number of adjacent permutations $\tau_1, \tau_2, \ldots, \tau_{d_{\tau} (\sigma, \phi)}$ so that $\phi = \sigma \circ \tau_1 \circ \tau_2 \circ \ldots \circ \tau_{d_{\tau} (\sigma, \phi)}$. 

Consider the graph whose vertices are the regions of the braid arrangement and such that two vertices share an edge if their corresponding regions are adjacent to each other (so that each edge corresponds to a hyperplane). Then, the Kendall tau distance is the shortest path distance of the graph. 

In technical terms: if $\sigma, \tau \in S_n$ where $\tau = (r,r+1)$ and $1$ is the identity in $S_n$, then
\[ d_{\tau} (1, \tau \circ \sigma) - d_{\tau} (1, \sigma) = \left\{\begin{matrix}
1 & \hspace{5pt} if \hspace{5pt} \sigma^{-1} (r) < \sigma^{-1} (r+1) \\ 
-1 & \hspace{5pt} if \hspace{5pt} \sigma^{-1} (r) > \sigma^{-1} (r+1) 
\end{matrix}\right. \]

We now show that the decoding distance behaves as a weighted version of the Kendall tau distance.  We remark that the function $f_i(\sigma) = \sum_{j=i+1}^n [ \sigma^{-1} (i) \leq \sigma^{-1} (j) ]$ used in Theorem \ref{f_i} to describe the function $S( \sigma )$ differs from the Kendall tau distance, since it considers not only the number of transpositions $( r,r+1)  $ but also the value of $r$: $f_r(\tau_r )=n-r-1$. 
As we shall see on the following theorem, diferentely from the Kendall tau distance where $|d_{\tau} (1, \tau \circ \sigma_r) - d_{\tau} (1, \sigma)|=1$ independently of $r$, the difference $|S (\tau_r \circ \sigma) - S (\sigma)|$ decreases with $r$.

\begin{myteo} \label{teo:trans}

Let $\sigma, \tau \in S_n$ where $\tau :=\tau_r= (r,r+1)$. Then,

\[ S (\tau \circ \sigma) - S (\sigma) = \left\{\begin{matrix}
- 2^{f_r(\sigma) - 1} & \hspace{5pt} if \hspace{5pt} \sigma^{-1} (r) < \sigma^{-1} (r+1) \\ 
2^{f_{r+1}(\sigma)} & \hspace{5pt} if \hspace{5pt} \sigma^{-1} (r) > \sigma^{-1} (r+1) 
\end{matrix}\right. \]
\end{myteo}

\begin{proof}
Since $(\tau \circ \sigma)^{-1}(r) = \sigma^{-1}(r+1)$ and
$(\tau \circ \sigma)^{-1}(j) = \left\{\begin{matrix}
\sigma^{-1}(r) & if \hspace{5pt} j=r+1 \\ 
\sigma^{-1}(j) & if \hspace{5pt} j > r+1
\end{matrix}\right.$
for $j \geq r+1$ it follows that 

\begin{align*}
f_r (\tau \circ \sigma) &= \sum_{j=r+1}^n [(\tau \circ \sigma)^{-1}(r) \leq (\tau \circ \sigma)^{-1}(j) ] \\
&= [\sigma^{-1}(r+1) \leq \sigma^{-1}(r)] + f_{r+1}(\sigma)
\end{align*}

Since $(\tau \circ \sigma)^{-1}(r+1) = \sigma^{-1}(r)$ and $r+1<j \Rightarrow (\tau \circ \sigma)^{-1}(j) = \sigma^{-1}(j)$
it follows that

\begin{align*}
\hspace{-50pt}
f_{r+1}(\tau \circ \sigma) =& \sum_{j=r+2}^n [(\tau \circ \sigma)^{-1}(r+1) \leq (\tau \circ \sigma)^{-1}(j) ] \\ &+ [ \sigma^{-1}(r) \leq \sigma^{-1}(r+1) ] - [ \sigma^{-1}(r) \leq \sigma^{-1}(r+1)] \\
=& f_r (\sigma) - [ \sigma^{-1}(r) \leq \sigma^{-1}(r+1)] 
\end{align*}

Thus, we have
\begin{align*}
\hspace{-20pt}
S (\tau \circ \sigma) &= \sum_{i = 1}^{r-1} 2^{f_i(\sigma)} + 2^{f_r(\tau \circ \sigma)} + 2^{f_{r+1}(\tau \circ \sigma)} + \sum_{i = r+2}^{n} 2^{f_i(\sigma)} \\
=&\hspace{6pt} S(\sigma) + 2^{f_r(\sigma)} (2^{-[\sigma^{-1}(r) \leq \sigma^{-1}(r+1)]} - 1) \\
&+ 2^{f_{r+1}(\sigma)} (2^{[\sigma^{-1}(r+1) \leq \sigma^{-1}(r)]} - 1).
\end{align*}

\end{proof}

As a direct corollary, using Corollary \ref{cor:decodingdist}:

\begin{mycor}
Let $\sigma, \tau \in S_n$ where $\tau = (r,r+1)$. Then,
\[ d_{dec} (1, \tau \circ \sigma) - d_{dec} (1, \sigma) = \left\{\begin{matrix}
\dfrac{2^{f_r(\sigma) - 1}}{2^n-1} & \hspace{5pt} if \hspace{5pt} \sigma^{-1} (r) < \sigma^{-1} (r+1) \\[10pt] 
- \dfrac{2^{f_{r+1}(\sigma)}}{2^n-1} & \hspace{5pt} if \hspace{5pt} \sigma^{-1} (r) > \sigma^{-1} (r+1) 
\end{matrix}\right. \]
\end{mycor}

\section{A Distance Between Stable Channels} \label{sec:discha}

In this section we extend the results of the previous one to define a distance (in some sense) between channels. As in the last section we will only consider the case where the decoding cones are full dimensional, i.e. the channel is stable.

We could define a distance by setting $d(P,Q) = 1 - Pr ( \hspace{3pt}  {D\widehat{e}c}_C (P) = {D\widehat{e}c}_C (Q)  \hspace{3pt}  )$, but we will see that a more refined distance can be defined.

Consider three channels $P,Q,R \in {Cha}_3$ such that
\[ O^- P = \begin{pmatrix}
1 & 3 & 3 \\ 
2 & 1 & 2 \\ 
3 & 2 & 1
\end{pmatrix} \hspace{7pt}
O^- Q = \begin{pmatrix}
1 & 3 & 2 \\ 
2 & 1 & 3 \\ 
3 & 2 & 1
\end{pmatrix} \hspace{7pt}
O^- R = \begin{pmatrix}
2 & 3 & 2 \\ 
1 & 1 & 3 \\ 
3 & 2 & 1
\end{pmatrix}
.\]

One can check, by doing all possible computations, that $d(P,Q) = d(P,R) = d(Q,R)= \frac{4}{7}$. But $Q$ differs from $P$ in only one position of a single column, while $R$ differs from $P$ in one position in two different columns. If $y_1$ or $y_2$ (the output messages corresponding, respectively, to the first and second columns) is received $P$ and $Q$ are essentially the same channel. Intuitively, we expect $Q$ to be closer to $P$ than $R$. 

This distance does not use the fact that the received message will be known at the time of decoding. We will use this fact to define a more refined distance.

If we assume that the transmission is made through the channel $P$, and denote by $Q_y$ the column corresponding to the received message $y$ in $Q$, we can calculate  $Pr ( \hspace{3pt}  {D\widehat{e}c}_C (P_y) = {D\widehat{e}c}_C (Q_y)\hspace{2pt}  )$, the probability that both decoders will be equal when a message $y$ is received\footnote{In this case both the code $C$ and the message $y$ are random variables.}. With this we can define the following distance:

\begin{mydef}\label{defi12}
Let $P,Q  \in {Ch}_{n\times  m}$ and assume that $P$ is the channel being used. The radial decoding distance to $Q$ centered in $P$ is given by
\[d_{dec}^P (Q) = 1 - Pr (\hspace{3pt} {D\widehat{e}c}_C (P_y) = {D\widehat{e}c}_C (Q_y) \hspace{3pt} )  .\]
\end{mydef}

The next theorem shows how to compute this distance.

\begin{myteo}\label{teo:cha}

Let $P, Q \in {Ch}_{n\times  m}$ and $\sigma, \phi \in S_n^m$ be such that $\sigma_i$ and $\phi_i$ correspond to the ordering in the $i$-th column of $O^- P$ and $O^- Q$, respectively. Suppose that the channel being used is $P$. If a code $C \subseteq X$ is picked uniformly distributed from the space of all codes, then

\[ Pr ( \hspace{3pt}  {D\widehat{e}c}_C (P_y) = {D\widehat{e}c}_C (Q_y)  \hspace{3pt}  ) =  \frac{1}{n(2^n - 1)}  \sum_{i = 1}^m S(\sigma_i, \phi_i) \left \Vert P_i \right \Vert_1  \]
where $\left \Vert P_i \right \Vert_1:=\sum_{j=1}^nP_{ji}$ is the $1$-norm of the $i$-th column of $P$.

\end{myteo}

\begin{proof}

\begin{align*}\hspace{-28pt}
Pr ( \hspace{3pt}  {D\widehat{e}c}_C (P_y) = {D\widehat{e}c}_C (Q_y)  \hspace{2pt}  )  &= \sum_{i = 1}^m Pr ( \hspace{2pt}  {D\widehat{e}c}_C (P_y) = {D\widehat{e}c}_C (Q_y) \vert \hspace{2pt} y_i \hspace{2pt} \text{received} ) Pr (y_i \hspace{2pt} \text{received}) \\
&= \sum_{i=1}^m \frac{S(\sigma_i, \phi_i)}{2^n - 1} \sum_{j=1}^n Pr (y_i \hspace{2pt} \text{received} \hspace{3pt} \vert x_j \hspace{2pt} \text{sent}) Pr(x_j \hspace{2pt} \text{sent}) \\
&= \sum_{i=1}^m \frac{S(\sigma_i, \phi_i)}{2^n - 1} \left \Vert P_i \right \Vert_1 \frac{1}{n}
\end{align*}

\end{proof}

In the hypothesis of Theorem \ref{teo:cha} we assume that one of the channels is the correct one. This occurs because the expression depends on the probability of receiving $y$ which may not coincide for different channels.

\begin{mycor} \label{cor:disrad}
Let $P,Q  \in {Ch}_{n\times  m}$ and assume that $P$ is the channel being used. The radial decoding distance to $Q$ centered in $P$ is given by
\[d_{dec}^P (Q) = 1 - \frac{1}{n(2^n - 1)}  \sum_{i = 1}^m S(\sigma_i, \phi_i) \left \Vert P_i \right \Vert_1  \]
\end{mycor}

We go back to the example discussed in the beggining of this section.

\begin{myexe} \label{exe:conta}

Suppose a channel $P = \begin{pmatrix}
\frac{5}{8} & \frac{1}{8} & \frac{2}{8} \\[1ex]
\frac{2}{8} & \frac{5}{8} & \frac{1}{8} \\[1ex]
\frac{1}{8} & \frac{2}{8} & \frac{5}{8}
\end{pmatrix}$
is used for transmission and $Q,R \in {Cha}_{3}$ are such that 
\[ O^- Q = \begin{pmatrix}
1 & 3 & 3 \\ 
2 & 1 & 2 \\ 
3 & 2 & 1
\end{pmatrix}
\hspace{5pt}
\text{and} 
\hspace{5pt}
O^- R = \begin{pmatrix}
	2 & 3 & 2 \\ 
	1 & 1 & 3 \\ 
	3 & 2 & 1
\end{pmatrix} .\]

Then, by Theorem \ref{teo:cha},
\begin{align*}
Pr ({D\widehat{e}c}_C (P_y) = {D\widehat{e}c}_C (Q_y)  &= \frac{1}{21} (7 + 7 + 4) = \frac{6}{7}
\end{align*} and
\begin{align*}
Pr ({D\widehat{e}c}_C (P_y) = {D\widehat{e}c}_C (R_y))  =  \frac{1}{21} (5 + 7 + 4) = \frac{16}{21}.
\end{align*}
Thus, $${d_{dec}^P(Q) = \dfrac{1}{7} < \dfrac{5}{21} = d_{dec}^P(R)}.$$

We note that this difference is, intuitively, compatible with the simple observation that $Q$ differs from $P$ in only one position of a single column, while $R$ differs from $P$ in one position in two different columns.
\end{myexe}

The decoding distance presented in Definition \ref{def:decdist} of the previous section was symmetric and only depended on the equivalence classes of the permutations. In contrast, the radial decoding distance to $Q$ centered in $P$ is not symmetrical, and although it only depends on the equivalence class of $Q$, it depends on the internal structure of $P$, i.e. if $Q \sim Q'$, then $d_{dec}^P(Q)=d_{dec}^P(Q')$, but $P\sim P'$ does not necessarily imply that $d_{dec}^P(Q)=d_{dec}^{P'}(Q)$.

\section{Discussion}

In this work, we gave an explicit expression for a meaningful distance in the space of all channels over given input and output sets. This establishes the ground to study the details of what can be a kind of finite
	approximation approach to channels and decodification problems. A family of
	questions that arise in this context are the following: Let $\overline
	{Ch_{n\times m}}=Ch_{n\times m} /\sim$ be the set of decoding cones and
	let $A\subset\overline{Ch_{n\times m}}$ be a subset of channels with some
	interesting property (for example, the set of channels that admits syndrome
	decoding). If we want to approximate a channel $P\in Ch_{n\times m}$ by a
	decoding cone in $A$, how much (in terms of decoding errors) should we expect
	to lose? From Corollary 3, we are actually interested in determining
	$\max\left\{  d_{dec}^{P}\left(  Q\right)  ;Q\in\overline{Ch_{n\times m}%
	}\right\}  $. Asymptotic versions arise naturally as we consider a family of
	increasing (in terms of $n=\left\vert X\right\vert $  and $m=\left\vert
	Y\right\vert $) input and output sets.  
	
	This approach is similar to the one adopted in the study of mismatched channels as, for example, in \cite{gant00}. The approach used in this (and other works studying mismatched channels) rests on the determination of achievable rates, that is, in proving that, for $n$ sufficiently large \textit{there are} codes that can be decoded with the approximating channel with no significant loss, that is, with probability of mis-decoding approaching $0$. In our approach we are not looking at this family of codes (asymptotically the  best choice of code for the mismatched channel), but on the average loss while choosing sequences of codes with a given rate.
	
	\bigskip
	
	We also stress that \textit{any} prescribed deterministic decision rule can be seen as a maximum likelihood decoding rule of some channel (actually an equivalence class of channels), as can be seen, for example, in \cite{fire14}.
	
	\bigskip
	
	Besides that, we remark that we considered the case of stable channels,
	i.e. the case of a decoding cone $cone(P)$ that is determined by a set of
	strict inequalities. An unstable (non-full dimensional) cone $Cone(P)$ is determined by a set of inequalities and a non-empty set of equalities, or, in other words, $Order(P)$ contains equivalences. It inherits its decoders from its full dimensional neighbours, that is, cones corresponding to stable channels in which every inequality of $Cone(P)$ also holds.

	Finding explicit expressions for a distance on the set of all decoding cones, both stable and unstable, is technically more challenging.

%
%

\bigskip

\renewcommand{\abstractname}{Acknowledgements:}
\begin{abstract}

Rafael G.L. D'Oliveira was supported by CAPES.
Marcelo Firer was partially supported by  S\~{a}o Paulo
Research Foundation, (FAPESP grant 2013/25977-7). 
\end{abstract}

\end{document}